\documentclass[runningheads]{llncs}
\usepackage{mdframed}
\usepackage[most]{tcolorbox}
\usepackage{todonotes}
\usepackage{mdframed}
\usepackage[most]{tcolorbox}
\usepackage{todonotes}
\usepackage[utf8]{inputenc} 
\usepackage[T1]{fontenc}    
\usepackage{hyperref}   
\usepackage{amsmath}
\usepackage{amssymb}
\usepackage{xspace}
\usepackage{mathtools}
\usepackage{graphicx,xcolor}
\usepackage{url}            
\usepackage{booktabs}       
\usepackage{amsfonts}       
\usepackage{nicefrac}       
\usepackage{microtype}      
\usepackage{lipsum}
\usepackage{fancyhdr}       
\usepackage{graphicx} 
\usepackage[ruled,linesnumbered]{algorithm2e}
\usepackage[utf8]{inputenc}
\graphicspath{{media/}}   

\newcommand{\xp}{{\sf XP}\xspace}
\newcommand{\fpt}{{\sf FPT}\xspace}

\newcommand{\np}{\textup{\textsf{NP}}\xspace}

\newcommand{\npc}{{\textsf{NP}\textrm{-complete}\xspace}}

\newcommand{\wth}{\textsf{W[2]}\textrm{-hard}\xspace}

\newcommand{\apxh}{\textrm{APX-hard}\xspace}

\newcommand{\mcs}{\textup{\textsc{MCS}}\xspace}
\newcommand{\mscs}{\textup{\textsc{MSCS}}\xspace}
\newcommand{\mcss}{\textup{\textsc{MCSS}}\xspace}
\newcommand{\mss}{\textup{\textsc{MSS}}\xspace}

\newcommand{\NN}{\mathrm{\hat{N}}\xspace}

\usepackage{multirow}

\usepackage{lineno}

\newtheorem{obs}{Observation}[theorem]
\usepackage{graphicx}
\begin{document}
\title{Minimum Selective Subset on Some Graph Classes}
%
%
\author{%
Bubai Manna
}%
\institute{
Indian Institute of Technology Kharagpur, India\\
\email{bubaimanna11@gmail.com}}

\authorrunning{B. Manna}
\maketitle              
\begin{abstract}
In a connected simple graph \( G = (V(G),E(G)) \), each vertex is assigned a color from the set of colors \( C=\{1, 2,\dots, c\} \). The set of vertices is partitioned as \( V(G) = \bigcup_{\ell=1}^{c} V_{\ell} \), where all vertices in \( V_{\ell} \) share the same color \( \ell \). A subset \( S \subseteq V(G) \) is called \textit{Selective Subset} if, for every vertex \( v \in V(G) \), if \( v\in V_{\ell} \), at least one of its nearest neighbors in \( S\cup (V(G)\setminus V_{\ell}) \) has the same color as \( v \). The \emph{Minimum Selective Subset} (\mss) problem seeks to find a selective subset of minimum size. The problem was first introduced by Wilfong in 1991~\cite{Wilfong} for a set of points in the Euclidean plane, where two major problems, \mcs (Minimum Consistent Subset) and \mss, were proposed.

In graph algorithms, the only known result is that the \mss problem is $\npc$, as shown in~\cite{BBC} in 2018. Beyond this, no further progress has been made to date. In contrast, the \mcs problem has been widely studied in various graph classes over the years. Therefore, in this work, we also extend the algorithmic study of \mss on various graph classes. We first present a \( \mathcal{O}(\log n) \)-approximation algorithm for general graphs with $n$ vertices and regardless of the number of colors. We also show that the problem remains $\npc$ even for planar graphs when restricted to just two colors. Finally, we provide linear-time algorithms for computing optimal solutions in trees and unit interval graphs for any number of colors\footnote{This work has been accepted in CCCG 2025}.

\end{abstract}

\section{Introduction}
In many supervised learning scenarios, we are given a set of labeled training data \( T \) in a metric space \( (X, d) \), where each point \( t \in T \) is assigned a \emph{color} (or class label) from a set \( C = \{1, 2, \dots, c\} \). A common goal is to reduce the dataset to a smaller representative subset without losing essential classification information. One such problem is the \textit{Minimum Consistent Subset} (\mcs), where the task is to find the smallest subset \( S \subseteq T \) such that every point in \( T \) is either in \( S \) or has a nearest neighbor in \( S \) with the same color. This problem was first introduced by Hart~\cite{Hart}, whose foundational work has been cited more than 2,800 times, reflecting its broad impact.

Unfortunately, the \mcs problem is computationally challenging. It has been shown to be \(\npc\) when the number of colors is three or more~\cite{Wilfong}, and it remains \(\npc\) even for just two colors in the plane~\cite{Khodamoradi}. The problem is also \(\text{W}[1]\)-hard when parameterized by the size of the output subset~\cite{Chitnis22}. Despite this, several algorithmic approaches have been developed to tackle \mcs in \(\mathbb{R}^2\)~\cite{BBC}\cite{ahmad}\cite{Chitnis22}\cite{Wilfong}, emphasizing its relevance in practical machine learning and geometric settings.

A closely related variant is the \textit{Minimum Selective Subset} (\mss) problem, also studied by Wilfong~\cite{Wilfong}. This problem aims to identify the smallest subset of training data that retains the necessary structure for classification. It is particularly useful in real-world applications like fingerprint recognition, handwritten character recognition, and other pattern recognition tasks, where eliminating redundancy can lead to more efficient and accurate models.

Wilfong~\cite{Wilfong} proved that \mss, like \mcs, is \(\npc\) even when the input is limited to two colors in \(\mathbb{R}^2\). More recently, Bhore et al.~\cite{BBC} presented a Polynomial-Time Approximation Scheme (PTAS) for the problem in the multi-color setting. They also showed that \mss is \(\text{W}[1]\)-hard when parameterized by the size of the solution and lies in the class \(\text{W}[\text{2}]\) for the special case of two colors in the plane. These results help position \mss within the broader landscape of parameterized complexity and approximation.

\subsection{Notations and Definitions}
For any graph \( G = (V(G),E(G)) \), we denote the set of vertices by \( V(G) \) and the set of edges by \( E(G) \). Without loss of generality, we use \( [n] \) to denote the set of integers \( \{1,\ldots, n\} \). We use an arbitrary vertex color function \( C: V(G)\rightarrow [c] \), such that each vertex is assigned exactly one color from the set \( [c] \). For a subset of vertices \( U \subseteq V(G)\), let \( C(U) \) represent the set of colors of the vertices in \( U \), formally defined as \( C(U) = \{C(u) \mid u \in U\} \). 

For any two vertices \( u,v \in V(G) \), the shortest path distance between \( u \) and \( v \) in \( G \) is denoted by \( d (u,v) \). \( d (u,v) \) is called \emph{hop-distance} between $u$ and $v$. For a vertex \( v \in V(G) \), the distance between \( v \) and the set \( U \subseteq V(G)\) in \( G \) is given by $d(v,U) = \min_{u\in U} d (v,u)$.

The nearest neighbors of \( v \) in the set \( U \) is denoted as \( \NN(v,U) \), formally defined as  
\[
\NN(v,U) = \{u \in U \mid d (v,u) = d (v,U) \}.
\]  
Therefore, if $v\in U$, then $\NN (v, U)=\{v\}$. 

\( G[U] \) denotes the subgraph of \( G \) induced by \( U \subseteq V(G)\), and \( \lvert U \rvert \) is the cardinality of $U$. We use standard graph-theoretic notation and symbols as presented in~\cite{diestel2012graph}.

Suppose \( G = (V(G), E(G)) \) is a given connected and undirected graph, where the vertices are partitioned into \( c \) color classes, namely \( V_1, V_2, \dots, V_c \). This means that each vertex in \( V(G) \) has a color from the set \( [c] \), and each vertex in \( V_{\ell} \) has color \( \ell \). 
A \emph{Minimum Consistent Subset} (\mcs) is a subset \( S \subseteq V(G) \) of minimum cardinality such that for every vertex \( v \in V(G) \), if \( v \in V_{\ell} \), then  
\[
\NN(v,S) \cap V_{\ell} \neq \emptyset.
\]  
The definition of a selective subset is as follows:
 \begin{definition}
A subset \( S \subseteq V(G) \) is called a \emph{Minimum Selective Subset (\mss)} if, for each vertex \( v \in V(G) \), if \( v \in V_{\ell} \), the set of nearest neighbors of \( v \) in \( S \cup (V(G) \setminus V_{\ell}) \), 
contains at least one vertex \( u \) such that \( C(v) = C(u) \), and \(\lvert S\rvert \) is minimum.

In other words, we are looking for a vertex set \( S \subseteq V(G) \) of minimum cardinality such that every vertex \( v \) has at least one nearest neighbor of the same color in the graph, excluding those vertices of the same color as \( v \) that are not in \( S \). 
\end{definition}

\begin{figure}[t]
\includegraphics[width=8cm]{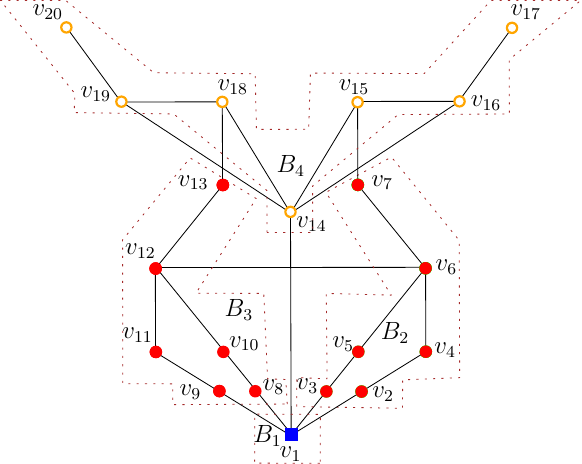}
\centering
\caption{ \textbf{Colors:} \emph{blue = square}, \emph{red = disk}, \emph{orange = fdisk}. The vertex set of the graph is given by \( V(G) = V_{\text{blue}} \cup V_{\text{red}} \cup V_{\text{orange}} \), where \( V_{\text{blue}} = \{v_1\} \), \( V_{\text{red}} = \{v_2, \dots, v_{13}\} \), and \( V_{\text{orange}} = \{v_{14}, \dots, v_{20}\} \). The sets \( \{v_1, v_4, v_5, v_{10}, v_{11}, v_{17}\} \) and \( \{v_1, v_3, v_4, v_8, v_{11}, v_{17}\} \) are both \mcs{} solutions. The sets \( M = \{v_1, v_2, v_3, v_7, v_8, v_9, v_{13}, v_{14}\} \) and \( M = \{v_1, v_4, v_5, v_7, v_{10}, v_{11}, v_{13}, v_{14}\} \) are both valid \mss{} solutions. Brown-dotted regions indicate the blocks. The complete list of blocks is: \( B_1 = \{v_1\} \), \( B_2 = \{v_2, \dots, v_7\} \), \( B_3 = \{v_8, \dots, v_{13}\} \), and \( B_4 = \{v_{14}, \dots, v_{20}\} \).}\label{bubai1}
\end{figure}
Figure~\ref{bubai1} illustrates that \mcs and \mss are distinct and may not be unique for a given graph. The selective subset problem on graphs is defined as follows:

\begin{tcolorbox}[
    enhanced,
    title={\color{black} \sc{Selective Subset Problem on Graphs}},
    colback=white,
    boxrule=0.4pt,
    attach boxed title to top center={xshift=0cm, yshift*=-2mm},
    boxed title style={size=small, frame hidden, colback=white}
]
    \textbf{Input:} A graph \( G=(V(G),E(G)) \), a color function \( C:V(G)\rightarrow [c] \), and an integer \(s \).\\
    \textbf{Question:} Does there exist a selective subset of size \( \leq s \) for \( (G,C) \)?
\end{tcolorbox}
Banerjee et al.~\cite{BBC} proved that the \mcs{} problem is \wth~\cite{book} when parameterized by the size of the minimum consistent subset, even for graphs with just two colors. This highlights the intrinsic difficulty of the problem, even in relatively simple settings. Building on this, Dey et al.~\cite{DeyMN21}\cite{10.1007/978-3-030-67899-9_37} \cite{DeyMN23} developed polynomial-time algorithms for \mcs{} on several restricted graph classes, including paths, spiders, caterpillars, combs, and trees (for trees, when the number of colors \( c = 2 \)). Additional algorithmic results---including \xp and \fpt algorithms (when \( c \) is treated as a parameter), as well as NP-completeness---have been obtained for trees~\cite{Arimura23} \cite{aritra}. Furthermore, \mcs is also $\npc$ in interval graphs \cite{aritra} and $\apxh$ in circle graphs \cite{manna2024minimum}.

Variants of the problem, such as the \textit{Minimum Consistent Spanning Subset} (\mcss) and the \textit{Minimum Strict Consistent Subset} (\mscs), have also been investigated on trees~\cite{biniaz2024minimum} \cite{manna2024minimumx}, enriching the broader study of consistency-based data selection problems.

In contrast, the algorithmic landscape of the \textit{Minimum Selective Subset} (\mss) problem remains largely unexplored. To date, the only known result is a general NP-completeness proof by Banerjee et al.~\cite{BBC}, which establishes that \mss{} is hard on general graphs. This reveals a significant gap in the literature and motivates a more detailed investigation into both the computational complexity and algorithmic tractability of the problem.

\subsection{Our Results}

While previous work has primarily focused on establishing the hardness of \mss, our study aims to bridge this gap by providing new insights into both its complexity and approximability across various graph classes. We identify several cases where efficient algorithms or approximations are achievable and others where hardness remains.

In Section~\ref{section3}, we present the first non-trivial approximation result for \mss{} in general graphs: an \( \mathcal{O}(\log n) \)-approximation algorithm, where \( n \) is the number of vertices. This result holds for any number of colors and offers a useful baseline for practical applications.

Planar graphs are among the most studied graph classes in algorithmic graph theory due to their rich structure and frequent appearance in real-world applications such as geographic information systems and circuit design. While many hard problems become easier on planar graphs, we show in Section~\ref{section4} that \mss{} remains \(\npc\) even in this restricted setting with only two colors. This illustrates the fundamental difficulty of the problem, even in well-behaved graph classes.

We then turn our focus to two more restricted and well-understood graph classes: trees and unit interval graphs. Trees are among the simplest graph structures and often admit elegant and efficient algorithms. Unit interval graphs, which model equal-length intervals on the real line, are widely used in areas such as scheduling and bioinformatics due to their tractable structure. These graph classes have already been studied in depth in the context of \mcs~\cite{aritra}, so it is natural to explore \mss{} in the same setting.

In Sections~\ref{section5} and ~\ref{section6}, we show that \mss{} can indeed be solved efficiently in both cases. Specifically, we provide \textit{linear-time algorithms} for computing optimal \mss{} solutions on trees and unit interval graphs for any number of colors. These results significantly extend the frontier of tractable cases for the \mss{} problem and open new directions for applications in structured domains.

\section{Preliminaries}\label{section2}
If all the vertices of a graph \( G \) are of the same color (i.e., \( G \) is monochromatic), then any single vertex forms a valid \mss{}, since every other vertex can select it as a nearest neighbor with the same color. 

More generally, in any graph \( G \), every selective subset must include at least one vertex from each color class. To see this, suppose for contradiction that there exists a color class \( V_\ell \subseteq V(G) \) such that \( V_\ell \cap S = \emptyset \), where \( S \) is a selective subset. Then, for any vertex \( v \in V_\ell \), all vertices in its nearest neighbor set \( \NN(v, S \cup (V(G) \setminus V_\ell)) \) must belong to a different color class (since \( S \cap V_\ell = \emptyset \)), contradicting the definition of a selective subset which requires each vertex to have a same-colored nearest neighbor in \( S \). 

\begin{definition}
A block is defined as a maximal connected subgraph whose vertices share the same color.
\end{definition}
An example illustrating the block decomposition is shown in Figure~\ref{bubai1}.
\begin{lemma}\label{bubailemma101}
Any selective subset must contain at least one vertex from each block. 
\end{lemma}
\begin{proof}
Suppose, for contradiction, that \( M \) is a selective subset of a graph \( G \), but there exists a block \( B_i \) such that \( M \cap B_i = \emptyset \). Let the color assigned to the vertices of \( B_i \) be \( \ell \), and let \( v \in B_i \) be any vertex.

Therefore, there must exist a nearest neighbor (say \( u \)) of \( v \) in \( M \cup (V(G) \setminus V_{\ell}) \) such that \( C(u) = C(v) \); otherwise, $M$ would not be a selective subset. Since \( M \cap B_i = \emptyset \), we have \( u \notin B_i \). Thus, \( u \in B_j \) for some \( j \neq i \).

Let \( P \) denote the shortest path between \( v \) and \( u \). As \( B_i \) and \( B_j \) are distinct blocks and \( C(v) = C(u) \), there must exist at least one vertex \( w \in P \) such that \( C(w) \neq C(u) \). This implies that \( w \) lies closer to \( v \) than \( u \) does, i.e., \( d(v,w) < d(v,u) \), and \( w \in M \cup (V(G) \setminus V_{\ell}) \) with \( C(w) \neq C(u) \).

Hence, instead of \( u \), $w$ is the nearest neighbor of \( v \) in \( M \cup (V(G) \setminus V_{\ell}) \) with $C(w)\neq C(v)$, violating the definition of a selective subset. This contradiction completes the proof.
\end{proof}
The above proof indicates that, for any vertex \( v \in B_i \), either \( v \) itself must belong to \( M \), or there must exist a vertex \( u \in B_i \) such that \( u \in M \) and \( u \) is the nearest neighbor of \( v \) in \( M \cup (V(G) \setminus V_{\ell}) \). In other words, each vertex \(v\in B_i \) must have its nearest neighbor in \( M \cup (V(G) \setminus V_{\ell}) \) that also lies within \( B_i \). Therefore, we have the following observation:
\begin{obs}\label{obs3}
The blocks are independent of each other in the solution of a selective subset. 
\end{obs}

\section{$\mathcal{O}(\log n)$-Approximation Algorithm of \mss in General Graphs}\label{section3}
\begin{figure}[t]
\includegraphics[width=10cm]{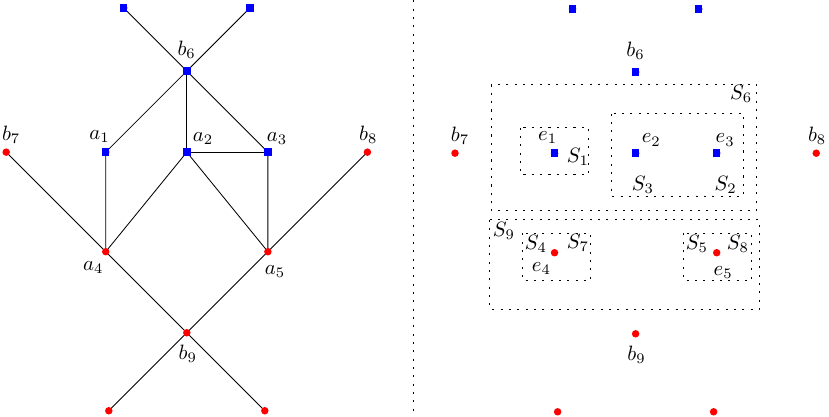}
\centering
\caption{Reduction to set cover. Colors: \emph{blue= square}, \emph{red= disk}. Sets are $B_{1,1}=\{a_1, a_2, a_3\}$, $B_{2,1}=\{a_4,a_5\}$, $B_{1,2}=\{b_6\}$, $B_{2,2}=\{b_7,$ $ b_8, b_9\}$. Sets are $B_{i,1}^{all}=\{a_1,$ $\dots,a_5\}$, $B^{all}=\{b_1,$ $\dots, b_9\}$, where $b_i=a_i$ for $i=1,\dots, 5$. The set elements for the set cover problem is $I=\{e_1,\dots, e_5\}$ where $e_i=b_i=a_i$ for $i=1,\dots,5$. The subsets are given $S_1=\{e_1\}$, $S_2=S_3=\{e_2,e_3\}$, $S_4=S_7=\{e_4\}$, $S_5=S_8=\{e_5\}$, $S_6=\{e_1, e_2, e_3\}$, $S_9=\{e_4, e_5\}$. The elements inside the dotted rectangles represent the sets. The selective subset is $M =\{b_6, b_9\}$, and the corresponding set cover is $\{S_6, S_9\}$.}\label{fig:1E} 
\end{figure}
Let \( G = (V(G), E(G)) \) be a graph with \( c \) colors, and let \( B_1, B_2, \dots, B_k \) denote the blocks in \( G \). The classical \emph{Set Cover} problem is defined as follows: Given a universe \( I \) and a collection \( \mathcal{S} \) of \( m \) subsets of \( I \), the objective is to select the smallest number of subsets from \( \mathcal{S} \) such that their union covers all elements of \( I \). We now describe a reduction from the \mss{} problem in \( G \) to an instance of the \textit{Set Cover} problem (see Figure~\ref{fig:1E}).

\textbf{Reduction.} Let \( B_{i,1} \subseteq B_i \) be the set of vertices adjacent to at least one vertex of a different color, for each \( i = 1, \dots, k \). For each vertex \( v \in B_i \setminus B_{i,1} \), if $v$ in adjacent to at least one vertex of $B_{i,1}$, include $v$ in $B_{i,2}$. Let \( B_{i,1}^{all} = \bigcup_{i=1}^{k} B_{i,1} = \{a_1, \dots, a_{n_1}\} \) and \( B^{all} = \bigcup_{i=1}^{k} (B_{i,1} \cup B_{i,2}) = \{b_1, \dots, b_{n_2}\} \) where $n_1\leq n$ and $n_2\leq n$. Therefore, each $b_j$ is a vertex of either $B_{i,1}$ or $B_{i,2}$ $\mbox{ for } j=1,\dots,n_2$ and $\mbox{for some } i\in \{1,\dots,k\}$. 

Construct a universe set \( I = \{e_1, \dots, e_{n_1}\} \), where each \( e_i \) corresponds to vertex \( a_i \in B_{i,1}^{all} \). For each \( b_i \in B^{all} \), define a set \( S_i \subseteq I \) such that \( e_j \in S_i \) if the corresponding vertex \( a_j \) and the vertex \( b_i \) are either the same vertex or adjacent and in the same block. Let \( S = \{S_1, \dots, S_{n_2}\} \).
\begin{lemma}{\label{bubailemma4}}
If a set cover of $I$ uses some sets of $S$, then the vertices represented by those sets form a selective subset, and vice versa.
\end{lemma}
\begin{proof}
Suppose $t$ sets in $S$ cover $I$, and let the corresponding vertices of these $t$ sets be $b_1$, $b_2$, \dots, $b_{t}$. We need to show that $M=\{b_1, \dots,b_{t}\}$ forms a selective subset of $G$.

Consider any vertex $v$ in the block $B_i$. Let $v\in V_{\ell}$ and $a_i\in B_{i,1}$ be the nearest vertex of $v$ in $B_{i,1}$. We have two cases as follows:
\begin{itemize}
    \item If $a_i\in M$, then $v$ has a nearest vertex in $M\cup (V(G)\backslash V_{\ell})$, which is $a_i$ and $C(a_i)=C(v)$.
    \item If $a_i \notin M$, then there must exist a vertex $b_i\in B_{i,2}$ adjacent to $a_i$. Since $\{S_1, \dots,S_{k}\}$ is a set cover of $I$, the corresponding element $e_i$ of the vertex $a_i$ must belong to at least one of $S_1, \dots,S_{k}$. This implies that $b_i$ must be in $M$ (as $a_i\notin M$); otherwise, $e_i$ would not be covered by any $S_x$, $1\leq x\leq t$, contradicting the assumption that $\{S_1, \dots,S_{t}\}$ is a set cover. Therefore, $v$ has a nearest neighbor in $M\cup (V(G)\backslash V_{\ell})$ which is $b_i$ and $C(v)=C(b_j)$, because $d (v,u) \geq d(v,b_i)$ for all $u\notin V_{\ell}$. 
\end{itemize}
This confirms that $M$ is indeed a selective subset.

Conversely, suppose the vertices corresponding to \( t \) sets in \( S \) form a selective subset of \( G \), and let those sets be \( S_1, \dots, S_t \). We need to show that \( S' = \{S_1, \dots, S_t\} \) is a set cover of \( I \).

Suppose, for contradiction, that $S'$ does not form a set cover of $I$. Then, there exists an element $e_j\in I$ such that $e_j\notin S_i$ for all $i=1,\dots,t$. Let the corresponding vertex of $e_j$ be $a_j\in B_{r,1}\subseteq B_r$. Also, let $D'$ denote the set of adjacent vertices of $a_j$ in $B_r$. Therefore, $D'\subseteq B^{all}$. Since $a_j \in B_{r,1}$, $a_j$ must have at least one adjacent vertex $a_x$ such that $C(a_j)\neq C(a_x)$.

Moreover, since $e_j\notin S_i$ $\mbox{ for } i=1,\dots,t$, no set corresponding to the vertices in $\{a_j\}\cup D'$ belongs to any of $S_i, \dots, S_t$; otherwise, $e_j$ would have appeared in at least one of these sets, contradicting our assumption that $S'$ is not a set cover of $I$. This implies that no vertex in $\{a_j\} \cup D'$ is contained in $M$. Suppose that $a_j\in V_{\ell}$. Consequently, $a_j$ does not have any nearest neighbor in $M\cup (V(G)\backslash V_{\ell})$, since $d(a_j, a_x)=1$ and $C(a_j)\neq C(a_x)$, leading to a contradiction. Thus, $S'$ must form a valid set cover of $I$.
\end{proof}

\begin{theorem}
The Minimum Selective Subset problem admits an \( \mathcal{O}(\log n) \)-approximation in general graphs.
\end{theorem}

\begin{proof}
By Lemma~\ref{bubailemma4}, the MSS problem reduces to a Set Cover instance where the universe size is \( n_1 \leq n \), with \( n = |V(G)| \).

The greedy algorithm for Set Cover achieves an \( \mathcal{O}(\log n_1) \)-approximation. Since \( n_1 \leq n \), it follows that $\log n_1 \leq \log n$,
and therefore the approximation factor becomes \( \mathcal{O}(\log n) \) in terms of the original graph size.

As the reduction preserves the approximation guarantee, this yields an \( \mathcal{O}(\log n) \)-approximation algorithm for the \mss problem in general graphs.
\end{proof}

\section{NP-Hardness of \mss in Planar Graphs}\label{section4}
We reduce an instance of the \textit{Planar Rectilinear Monotone 3-SAT} (PRM-3SAT) problem, denoted by \(\theta\), to a bichromatic planar graph, which forms the basis of our hardness construction. To set the stage, we begin with a formal definition of PRM-3SAT.

\begin{definition}
A \emph{Planar Rectilinear Monotone 3-SAT} (PRM-3SAT) formula \(\theta\) is a constrained version of the classical 3-SAT problem, where each clause contains either only positive literals or only negative literals (i.e., the formula is \emph{monotone}). The formula \(\theta\) can be embedded in the plane using a rectilinear planar layout \(\xi\) in \(\mathbb{R}^2\), which satisfies the following properties (see Figure \ref{fig:PRM}):
\begin{enumerate}
    \item Variables and clauses are represented in $\xi$ as axis-aligned squares and rectangles, respectively.
    \item All squares corresponding to variables are of equal size and are placed along the $x$-axis.
    \item All clause rectangles have the same height, although their widths may vary.
    \item Rectangles for positive clauses are placed above the $x$-axis, while those for negative clauses are placed below.
    \item Variables are connected to their associated clauses by vertical line segments, called \emph{clause-variable connecting paths} (CVC-paths).
    \item The endpoints of all CVC-paths lie on lattice points in the embedding $\xi$.
\end{enumerate}
\end{definition}

 \begin{figure}[t]
\centering
\includegraphics[width=11cm]{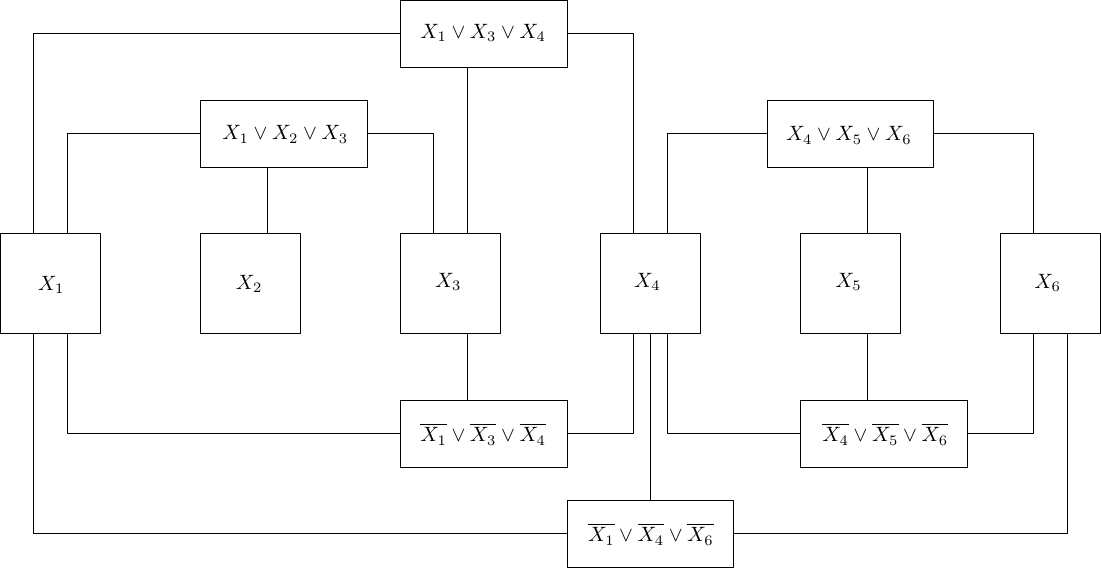}
\caption{\label{fig:PRM} An example of Planar Rectilinear Monotone 3SAT (PRM-3SAT)}
\end{figure}

The work of~\cite{PRM-3SAT} shows that PRM-3SAT is \(\npc\), even under these geometric and structural restrictions. Moreover, it is also shown that, given a PRM-3SAT formula \(\theta\), one can construct the corresponding rectilinear planar embedding \(\xi\) in polynomial time. This makes PRM-3SAT an ideal starting point for reductions to geometric and graph-theoretic problems involving planarity and monotonicity.

 \begin{figure}[t]
\centering
\includegraphics[width=11cm]{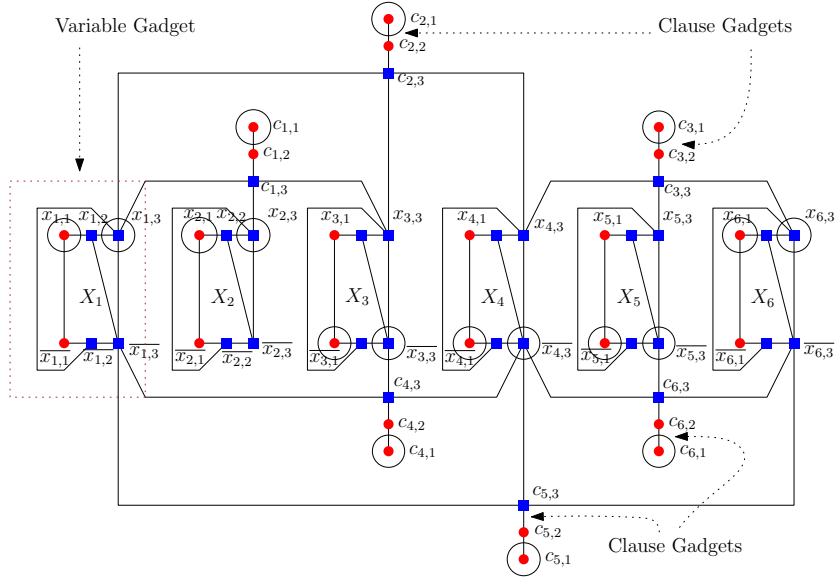}
\caption{\label{fig:1L} Reduction from PRM-3SAT to \mss in planar graph. Colors: \emph{blue= square}, \emph{red= disk}. The SAT expression is $\theta=c_1\wedge c_2\wedge c_3\wedge c_4\wedge c_5\wedge c_6$. The clauses are $c_1=(x_1\vee x_2\vee x_3), c_2=(x_1\vee x_3\vee x_4), c_3=(x_4\vee x_5\vee x_6), c_4=(\overline{x_1}\vee \overline{x_{3}}\vee \overline{x_4}), c_5=(\overline{x_1}\vee \overline{x_4}\vee\overline{x_6}), c_6=(\overline{x_4}\vee \overline{x_5}\vee\overline{x_6})$ with variables $x_1=x_2=x_6=1$ and $x_3=x_4=x_5=0$. The vertices inside small circles are in the selective subset.}
\end{figure}

\textbf{Reduction.} We embed an instance of PRM-3SAT formula $\theta$ with $n$ variables ${x_1, x_2, \dots , x_n}$ and $m$ clauses ${c_1, c_2, \dots, c_m}$ into a bichromatic planar graph $G=(V(G),E(G))$ (see Figure \ref{fig:1L}).

\textbf{Variable gadget:} For each variable \( x_i \) ($1\leq i\leq n$), we construct a variable gadget \( X_i \) shown as a brown dotted rectangle in Figure~\ref{fig:1L}, as follows: A path \( (x_{i,1}, x_{i,2}, x_{i,3}) \) of length two, where \( C(x_{i,1}) = \emph{red} \) and \( C(x_{i,2}) = C(x_{i,3}) = \emph{blue} \), is called the \emph{positive literal path} of \( x_i \). Similarly, a path \( (\overline{x_{i,1}}, \overline{x_{i,2}}, \overline{x_{i,3}}) \) of length two, where \( C(\overline{x_{i,1}}) = \emph{red} \) and \( C(\overline{x_{i,2}}) = C(\overline{x_{i,3}}) = \emph{blue} \), is called the \emph{negative literal path} of \( x_i \).

The vertex \( x_{i,1} \) is adjacent to \( \overline{x_{i,1}} \), and \( x_{i,2} \), \( x_{i,3} \) are adjacent to \( \overline{x_{i,3}} \), \( \overline{x_{i,2}} \), respectively. Additionally, \( x_{i,3} \) is adjacent to \( \overline{x_{i,3}} \).

\textbf{Clause gadget:} For each clause $c_j$ ($1\leq j\leq m$), the clause gadget $C_j$ is as follows: A path $(c_{j,1}, c_{j,2}, c_{j,3})$ of length two, where $C(c_{j,1})=C(c_{j,2})=$ \emph{red} and $C(c_{j,3})=$ \emph{blue}, is called \emph{clause path} for the clause $c_j$.

If a clause \( c_j \) consists of three positive literals \( x_i, x_l, x_t \) (i.e., \( c_j = (x_i \vee x_l \vee x_t) \)), then the vertex \( c_{j,3} \) adjacent to \( x_{i,3} \), \( x_{l,3} \), and \( x_{t,3} \) from their respective positive literal paths. If \( c_j \) has three negative literals \( \overline{x_i}, \overline{x_l}, \overline{x_t} \), then \( c_{j,3} \) is adjacent to \( \overline{x_{i,3}}, \overline{x_{l,3}}, \overline{x_{t,3}} \) from the corresponding negative paths as shown in the Figure \ref{fig:1L}.

The construction of the graph $G$ is now complete. Notably, $G$ remains planar because its embedding (Figure \ref{fig:1L}) closely resembles the PRM-3SAT embedding described in the Figure ~\ref{fig:PRM}. Thus, for $n$ variables and $m$ clauses, the bichromatic planar graph $G$ contains $6n+3m$ vertices and $8n+5m$ edges. We set $V(G)=V_{r}\cup V_{b}$, where $V_r$ is the set of \emph{red} vertices and $V_b$ is the set of \emph{blue} vertices.

\begin{lemma}{\label{mannalemma24}}
     $\theta$ is satisfied if and only if $G$ has a selective subset of size $2n+m$. 
\end{lemma}
\begin{proof}
Suppose $\theta$ is satisfiable. We construct a selective subset $M$ of size $2n+m$. Initially, $M:=\emptyset$. For each variable $x_i$, either $x_1=1$ or $\overline{x_i}=1$. If $x_i=1$, then include $x_{i,1}$ and $x_{i,3}$ in $M$; otherwise, $\overline{x_{i,1}}, \overline{x_{i,3}}\in M$. Thus, for $n$ variables, $M$ contains $2n$ vertices. Additionally, for each clause $c_j$, $M$ includes $c_{j,1}$. Therefore, $\lvert M\rvert =2n+m$.

Now we show that $M$ is a selective subset. In the variable gadget $X_i$, if $x_{i,1}, x_{i,3}\in S$, then $x_{i,2}$, $\overline{x_{i,2}}$ and $\overline{x_{i,3}}$ all have $x_{i,3}$ as their nearest neighbor of the same color in $M\cup (V(G)\backslash V_b)$ because $d(x_{i,2},x_{i,3})=d(\overline{x_{i,2}},x_{i,3})=d(\overline{x_{i,3}},x_{i,3})=1$, and $C(x_{i,2})=C(\overline{x_{i,2}})=C(\overline{x_{i,3}})=C(x_{i,3})$. Meanwhile $\overline{x_{i,1}}$ has $x_{i,1}$ as its nearest neighbor of the same color in $M\cup (V(G)\backslash V_r)$ because $d(\overline{x_{i,1}},x_{i,1})=1$ and $C(\overline{x_{i,1}})=C(x_{i,1})$. Similarly, if $\overline{x_{i,1}}, \overline{x_{i,3}}\in M$ then $x_{i,2}$, $x_{i,3}$ and $\overline{x_{i,2}}$ all have $\overline{x_{i,3}}$ as their nearest neighbor of the same color in $M\cup (V(G)\backslash V_b)$, and $x_{i,1}$ has $\overline{x_{i,1}}$ as its nearest neighbor of the same color in $M\cup (V(G)\backslash V_r)$. 

For each clause gadget $C_j$, since $c_{j,1}\in M$, the vertex $c_{j,2}$ has $c_{j,1}$ as its nearest neighbor of the same color in $M\cup (V(G)\backslash V_r)$. Now, we verify that at least one nearest neighbor of $c_{j,3}$ in $M\cup (V(G)\backslash V_b)$ has the same color. Suppose that $c_j=(x_i\vee x_l\vee x_t)$. Since each $c_j=1$ (because $\theta$ is satisfiable), at least one of $x_{i,3}$, $x_{l,3}$ and $x_{t,3}$ must be in $M$, ensuring that at least one nearest neighbor of $c_{j,3}$ in $M\cup (V(G)\backslash V_b)$ has color $C(c_{j,3})$ because $d(c_{j,3},x_{i,3})=d(c_{j,3},x_{l,3})=d(c_{j,3},x_{t,3})=1$ and $C(c_{j,3})=C(x_{i,3})=C(x_{l,3})=C(x_{t,3})$. Hence, $M$ is a selective subset of size $2n+m$.

Conversely, suppose $M$ is a selective subset of size $2n+m$. We determine which vertices are in $M$. Since $\{x_{i,1}, \overline{x_{i,1}}\}$ is a block, Lemma \ref{bubailemma101} implies that at least one vertex of the block must be in $M$. If $x_{i,1}\in M$, then at least one of $x_{i,2}$, $x_{i,3}$ $\overline{x_{i,2}}$, $\overline{x_{i,3}}$ must be in $M$; otherwise, $x_{i,2}$ and $\overline{x_{i,2}}$ will not have a nearest neighbor of their own color in $M\cup (V(G)\backslash V_b)$. We include $x_{i,3}$ in $M$ because $c_{j,3}$ (if $c_j=(x_i\vee x_l\vee x_t)$) will have $x_{i,3}$ as its nearest neighbor of the same color in $M\cup (V(G)\backslash V_b)$. Similarly, if $\overline{x_{i,1}}\in M$, then $\overline{x_{i,3}}\in M$. If we choose $x_{i,1}$ and $x_{i,3}$ in $M$, then put $x_i=1$; otherwise, $x_i=0$. This ensure that for $n$ variable gadgets, $M$ contains $2n$ vertices. 

For each clause $c_j$, the set $\{c_{j,1}, c_{j,2}\}$ forms a block and by Lemma \ref{bubailemma101}, at least one vertex of the block $\{c_{j,1}, c_{j,2}\}$ must be in $M$. We include $c_{j,1}\in M$, contributing one vertex per clause, thus $m$ vertices in total. Since $\lvert M\rvert=2n+m$, no other vertices are included. Note that $c_{j,2}$ has $c_{j,1}$ as its nearest neighbor of the same color in $M\cup (V(G)\setminus V_{r})$,

Now consider a clause $c_j=(x_i\vee x_l\vee x_t)$. If none of the pairs $\{x_{i,1}, x_{i,3}\}$, $\{x_{l,1}, x_{l,3}\}$, or $\{x_{t,1}, x_{t,3}\}$ are in $M$, then $c_{j,2}$ becomes the nearest neighbor of $c_{j,3}$ in $M\cup (V(G)\backslash V_b)$, violating the definition of selective subset because $C(c_{j,3})\neq C(c_{j,2})$. Thus, at least one of these pairs must be included in $M$. If $\{x_{i,1}, x_{i,3}\}\in M$, assign $x_i=1$; otherwise, assign $\overline{x_i}=1$ that we have already assigned. The same argument applies to $\{x_{l,1}, x_{l,3}\}$ and $\{x_{t,1}, x_{t,3}\}$. This ensures that at least one literal in each clause is true, proving that $\theta$ is satisfiable.
\end{proof}

\begin{theorem}
Finding a minimum selective subset is $\npc$ for planar graphs with two colors.
\end{theorem}
\begin{proof}
It is easy to see that the problem is in \np. As for $\npc$, Lemmas \ref{mannalemma24} establishes a relationship between $\theta$ and the size of the selective subset of $G$ in polynomial time. Therefore, \mss is $\npc$ in planar graphs.
\end{proof}
\textbf{Remark.} The above reduction remains valid even if the pair of vertices \( \{c_{j,1}, c_{j,2}\} \) in each clause gadget is assigned a distinct (except blue color), unique color not shared across gadgets. That is, we may assign a different color to each such pair for \( 1 \leq j \leq m \), and the reduction still preserves the equivalence between satisfying the PRM-3SAT formula \( \theta \) and the existence of a selective subset of size \( 2n + m \).

\section{Linear-time Algorithm of \mss in Trees}\label{section5}
We now describe a linear-time algorithm for computing a \mss{} in a tree. The key idea relies on Observation~\ref{obs3}, which states that the problem can be solved independently within each block. This significantly simplifies the structure of the problem and allows us to focus on local solutions. Therefore, to streamline both notation and presentation, we first consider the problem restricted to a single block.

Let \( T = (V(T), E(T)) \) be a tree rooted at a vertex \( r \), with \( |V(T)| = n \) vertices and a total of \( c \) colors. The tree is assumed to be connected and acyclic, as is standard. For each block \( B \subseteq V(T) \), we compute a minimum selective subset \( M_B \) using the strategy outlined in Lemma~\ref{bubailemma101}. The intuition is that within each block, local consistency can be enforced independently, and by handling each block separately, we maintain correctness globally.

The overall algorithm proceeds in two main phases. The first phase is \emph{initialization}, described in Algorithm~\ref{alg:0x}, where we preprocess the tree, organize its structure for efficient traversal, and gather relevant information needed for selection. The second phase is \emph{selection}, given in Algorithm~\ref{alg:02}, where we apply a bottom-up dynamic programming approach to determine the optimal selective subset within each block. 

Finally, the complete solution to the \mss{} problem on the entire tree is obtained by taking the union of the selective subsets computed for all the blocks. This modular approach not only ensures linear time complexity but also makes the algorithm scalable and easy to implement.

\begin{algorithm}
\caption{Initialization Phase for Computing a Minimum Selective Subset in Trees}\label{alg:0x}
\KwIn{A tree $T=(V(T), E(T))$ with $\lvert V(T)\rvert =n$ and $c$ colors. List of blocks $B_1, B_2, \dots, B_k$.}
\KwOut{Important Sets $B^{1}$, $B^{2}$, $B^{all}$, for each block $B$.}

\textbf{Initialization:} \\
\For{each block $B$}{

initialize $B^{1} \gets \emptyset$ and $B^{2} \gets \emptyset$.

\BlankLine

\textbf{Phase 1: Identifying Important Vertices} \\
\For{each vertex $v \in B$}{
    \If{$v$ is adjacent to at least one vertex of a different color}{
        $B^{1} \gets B^{1} \cup \{v\}$;
    }
}

\For{each vertex $u \in B\setminus B^1$}{
    \If{$u$ is adjacent to at least one vertex in $B^{1}$}{
        $B^{2} \gets B^{2} \cup \{u\}$;
    }
}

\BlankLine

Define:
    $B^{all} \gets B^{1} \cup B^{2}$;
}

\BlankLine

\end{algorithm}
The algorithm follows these steps (see Figure ~\ref{fig:1H}):

    \textbf{Initialization} (Algorithm \ref{alg:0x}).
    \begin{itemize}
        \item Start with the empty sets: $M_B:=\emptyset$, $B^1:=\emptyset$, and $B^2:=\emptyset$.
        \item For each vertex $v\in B$, if $v$ is adjacent to a vertex of a different color, include $v$ in $B^1$.
        \item For each vertex $v\in B\backslash B^1$, if $v$ is adjacent to at least one vertex in $B^1$, then include $v$ in $B^2$.
        \item Define \( B^{all} = B^1 \cup B^2 \) (note that \( B^1 \cap B^2 = \emptyset \)).
    \end{itemize}
   \textbf{Subtree Formation} (refer to lines 1-4 in Algorithm \ref{alg:02}).
      \begin{itemize}
          \item Since $B^{all}\subseteq B$, $B^{all}$ induces one or more connected induced subtrees in $T$, and each subtree consists of vertices of the same color because $B^{all}\subseteq B$. We consider such induced connected maximal subtrees.
          \item Let $T^{v_{1}}$, $T^{v_{2}}$, \dots, $T^{v_{t}}$ be such connected maximal subtrees in $T$ formed by the vertices of $B^{all}$ and call their roots $v_{1}, v_{2}, \ldots, v_{t}$, respectively. Note that, since the whole tree is rooted, a root is naturally defined for each subtree.
      \end{itemize}
  \textbf{Selection Process} (refer to lines 5-31 in Algorithm \ref{alg:02}).
       \begin{itemize}
       \item For each subtree \( T^{v_{j}} \) (where \( 1 \leq j \leq t \)), start from the \textit{lowest-level vertex} \( u \). A lowest-level vertex of a tree is a vertex that is farthest from the root.
       \item $u$ must belong to either $B^1$ or $B^2$.
       \item If $u\in B^2$:
           \begin{itemize}
               \item Remove $u$ from $B^2$, $B^{all}$, and $T^{v_{j}}$, as it must be adjacent to a vertex in $B^1$. 
           \end{itemize}
       
       \item If $u\in B^1$:
           \begin{itemize}
               \item If $u$ has a parent (say $v$) in $T^{v_{j}}$, add $v$ to $M_B$ and remove $v$ along its children from $B^1$, $B^2$, $B^{all}$, and $T^{v_{j}}$.
               \item If $u$ has no parent in $T^{v_{j}}$, add $u$ to $M_B$ and remove it from $B^1$, $B^{all}$, and $T^{v_{j}}$.
               \item If the grandparent of \( u \) (if it exists in $T^{v_{j}}$) belongs to \( B^1 \), move it from \( B^1\) to \( B^2 \), and update \( B^{all} \) accordingly.
           \end{itemize}
        \item Repeat the \textbf{Selection Process} until \( T^{v_{j}} \) becomes empty. Once \( T^{v_{j}} \) has no vertices, increase \( j \leftarrow j+1 \) and repeat the \textbf{Selection Process} until \( j = t \).
           \end{itemize}
    We apply the algorithm to each block \( B \) and obtain subsets \( M_B \). The final solution is \( M = \bigcup_{B \in T} M_B \). We now prove that \( M \) is not just selective, but also minimum.

\begin{algorithm}
\caption{Algorithm for Minimum Selective Subset of Trees}\label{alg:02}
\KwIn{A tree \( T = (V(T), E(T)) \) with \( \lvert V(T) \rvert = n \) vertices and a total of \( c \) colors. List of blocks \( B_1, B_2, \dots, B_k \). The subsets generated by Algorithm~\ref{alg:0x} are \( B^{1} \), \( B^{2} \), and \( B^{all} \) for each block $B$.}
\KwOut{A Minimum Selective Subset $M$.}

\textbf{subtree Formation:} \\

\For{each block $B$}{
    Identify connected maximal subtrees within the vertex of $B^{all}$:
    \[
    T^{v_{1}}, T^{v_{2}}, \dots, T^{v_{t}}
    \]
    where each subtree is rooted at $v_{1}, v_{2}, \dots, v_{t}$.
}

\BlankLine

\textbf{Selection Process: (Constructing the Selective Subset):} \\

\For{each block $B$}{
Set $j \gets 1$; $M_B \gets \emptyset$; \\
\While{$j\leq t$}{

    \While{$T^{v_{j}}$ has at least one vertex}{
        Select a lowest level vertex $u$ of $T^{v_j}$;

        \If{$u \in B^2$}{
            Remove $u$ from $B^{2}$, $B^{all}$, and $T^{v_{j}}$;
        }
        \Else{
            \If{$u$ has a parent (say $v$) in $T^{v_{j}}$}{
                Add the $v$ to $M_B$;
                
                Remove $v$ along with its children from $B^{1}$, $B^{2}$, $B^{all}$, and $T^{v_{j}}$;
            }
            \Else{
                Add $u$ to $M_B$;

                Remove $u$ from $B^{1}$, $B^{all}$, and $T^{v_{j}}$;
            }
            \If{the grandparent of $u$ (if exists) is in $B^{1}$}{
                Move it from $B^{1}$ to \( B^{2} \), and update \( B^{all} \) accordingly.
            }
        }
    }
    $j\gets j+1;$
}
\BlankLine
}
\Return $M\gets \bigcup_{B\in T}M_B$;
\end{algorithm}


\begin{figure}[t]
\centering
\includegraphics[width=10cm]{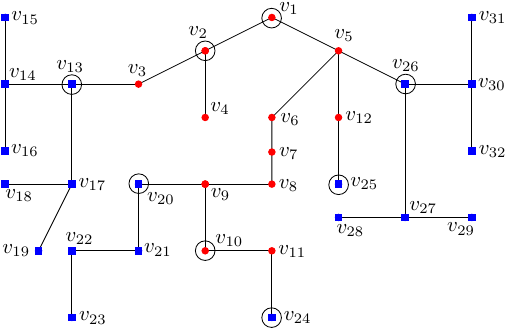}
\caption{Colors: \emph{blue= square}, \emph{red= disk}. $r=$ $v_1$ is the root of the tree $T$. The blocks are $B_1=\{v_1,$ $\dots,$ $ v_{12}\}$, $B_2=$ $\{v_{13},$ $\dots, $ $v_{19}\}$, $B_{3}=$ $\{v_{20},$ $\dots,$ $v_{23}\}$, $B_{4}$ $=\{v_{24}\}$, $B_{5}=$ $\{v_{25}\}$, $B_{6}$ $=\{v_{26}$ $,\dots,$ $v_{32}\}$. $B_1^{1}$ $=\{v_3,$ $ v_{5}, $ $v_{9}, $ $v_{11}, $ $v_{12}\}$, $B_1^{2}=\{v_1, v_{2}, v_{6}, v_{8}, v_{10}\}$,  $B_2^{1}=\{v_{13}\}$, $B_2^{2}=\{v_{14},v_{17}\}$, $B_3^{1}=\{v_{20}\}$, $B_3^{2}=\{v_{21}\}$, $B_4^{1}=\{v_{24}\}$, $B_4^{2}=\emptyset$, $B_5^{1}=\{v_{25}\}$, $B_5^{2}=\emptyset$, $B_6^{1}=\{v_{26}\}$, $B_6^{2}=\{v_{27},v_{30}\}$. $T_{1}^{v_1}=\{v_1,v_2,v_3,v_{5},v_{6},v_{12}\}$, $T_{1}^{v_{8}}=\{v_{8},v_{9}, v_{10},v_{11}\}$, $T_{2}^{v_{13}}=\{v_{13},v_{14}, v_{17}\}$, $T_{3}^{v_{20}}=\{v_{20},v_{21}\}$, $T_{4}^{v_{24}}=\{v_{24}\}$, $T_{5}^{v_{25}}=\{v_{25}\}$, $T_{6}^{v_{26}}=\{v_{26},v_{27},v_{30}\}$. $\mss =\{v_2, v_{1},v_{10},v_{13}, v_{20}, v_{24}, v_{25},v_{26}\}$. The vertices inside small circles are in the minimum selective subset.}\label{fig:1H}
\end{figure}
\begin{lemma}\label{treelemma}
For any block $B$ and any vertex \( u \in B^{1} \), either \( u \in M \) or at least one adjacent vertex of \( u \) in \( B^{\text{all}} \) must belong to \( M \).
\end{lemma}

\begin{proof}
Let \( X \) be the set that contains \( u \) and all its adjacent vertices in \( B^{\text{all}} \). Suppose \( T^{v_{j}} \) is the subtree that contains the vertices of \( X \). 

By the process described in Algorithm~\ref{alg:02}, after a finite number of iterations, a vertex from \( X \) will become a lowest-level vertex of \( T^{v_{j}} \) (since we delete vertices in each iteration until the subtree becomes empty).

We analyze two cases:

\begin{itemize}
    \item If \( u \) becomes a lowest-level vertex, then its parent (if it exists in \( T^{v_{j}} \)) must be included in \( M_B \) because $u\in B^1$, meaning that an adjacent vertex of \( u \) from \( B^{all} \) is included in \( M_B \). If the parent does not exist, then \( u \) itself must be included in \( M_B \) according to Algorithm~\ref{alg:02} as $u\in B^1$.
    
    \item If a vertex \( v \in X \setminus \{u\} \) becomes a lowest-level vertex, we consider two subcases:
    \begin{itemize}
    \item If \( v \in B^1 \), then its parent must be \( u \), as \( u \) is the only adjacent vertex of \( v \) (as $v$ is lowest-level vertex), and \( u \) must be included in \( M_B \).
    \item If \( v \in B^2 \), then \( v \) is removed from the tree. In subsequent iterations, either \( u \) eventually becomes a lowest-level vertex (in which case we proceed as above), or one of its remaining neighbors becomes a lowest-level vertex in \( B^1 \), again forcing \( u \in M_B \). If all neighbors of \( u \) are in \( B^2 \) and get deleted, \( u \) itself becomes lowest-level, completing the process.
\end{itemize}

\end{itemize}

Combining all these cases, we conclude that \( X \cap M_B \neq \emptyset \). Thus, either \( u \in M \) or at least one adjacent vertex of \( u \) in \( B^{all} \) belongs to \( M \).
\end{proof}
%
%
%
\begin{lemma}\label{tree}
$M$ is a minimum selective subset of the tree $T$.
\end{lemma}
\begin{proof}
We prove this in two steps: First, we show that \( M \) is a selective subset. Then, we prove that \( M \) is of minimum size among all selective subsets.
    
\textbf{Step 1:} (Selective subset)  
Suppose, for contradiction, that \( M \) is not a selective subset. Then there exists a vertex \( x \in B \) (let \( x \in V_{\ell} \)) such that \( x \) has no nearest neighbor of the same color in \( M \cup (V(T) \setminus V_{\ell}) \). 

Let \( D \) be a block adjacent to \( B \), and let \( y \in D \) be a nearest neighbor of \( x \) in $D$. Let \( P = x v_1 v_2 \dots v_l y \) be the shortest path from \( x \) to \( y \), where \( v_1, \dots, v_l \in B \) and \( y \in D \). In particular, \( v_l \in B^{1} \). 

By Lemma \ref{treelemma}, either \( v_l \in M_{B} \) or at least one neighbor of \( v_l \) in \( B^{all} \) belongs to \( M_{B} \).  
\begin{itemize}
    \item If \( v_l \in M_{B} \), then \( v_l \) must be the nearest neighbor of \( x \) in \( M \cup (V(T) \setminus V_{\ell}) \), contradicting our assumption. 
    \item If an adjacent vertex \( z \in B^{all} \) of \( v_l \) belongs to \( M_{B} \), then \( d(x, z) \leq d(x, y) \). Since \( z \in M \cup (V(T) \setminus V_{\ell}) \) and \( C(z) = C(x) \), it again contradicts our assumption that \( x \) has no such nearest neighbor.
\end{itemize}

Thus, no such \( x \) exists, and \( M \) is indeed a selective subset.

\textbf{Step 2:} (Minimality)  
We now show that \( M \) is a minimum selective subset.

Consider a subtree \( T^{v_{j}} \) rooted at \( v_{j} \) in the block $B$, and let \( u \) be a lowest-level vertex in \( T^{v_{j}} \). If \( u \in B^2 \), then its unique parent \( v \) must belong to \( B^1 \), since \( u \in B^{2} \) and \( v \) is its only adjacent vertex in \( T^{v_{j}} \) by construction. Instead of including \( u \) in \( M_B \), we defer selection to either \( v \) or its parent \( w \) (if it exists) in subsequent iterations, for the following reason:

If \( v \) becomes the lowest-level vertex and we add its parent \( w \) to \( M_B \), then \( u \) satisfies \( d(u, w) \leq d(u, v') \), where \( v' \) is a vertex adjacent to \( v \) of a different color, and \( C(u) = C(v) \). This satisfies the selective subset condition for \( u \). A similar argument holds if \( w \) does not exist and we instead add \( v \) to \( M_B \). In either case, including \( u \) in \( M_B \) is unnecessary. Therefore, we remove \( u \) from all relevant structures: \( B^2 \), \( B^{all} \), and \( T^{v_{j}} \).

Now consider \( u \in B^{1} \). By Lemma~\ref{treelemma}, either \( u \in M_B \), or at least one of its adjacent vertices in \( B^{all} \) must be included in \( M_B \), satisfying the selective subset condition for \( u \).

Suppose, for contradiction, removing a vertex \( u \) from \( M \) still leaves \( M \) as a valid selective subset. Let \( u \) have been added to \( M \) during a specific iteration of the algorithm on subtree \( T^{v_{j}} \). We analyze the following cases:

\begin{itemize}
    \item If \( u \) was a lowest-level vertex at the time of inclusion in \( M_B \), then \( u \notin B^{2} \); otherwise, it would have been excluded. Thus, \( u \in B^{1} \), and by construction, \( u \) had no parent in \( T^{v_{j}} \); otherwise, the parent could have been selected instead. If we now exclude \( u \) from \( M_B \), it has no adjacent vertex in \( M_B \cap B^{all} \) with the same color, violating the selective subset condition—a contradiction.

    \item If, in a later iteration, \( u \) became the parent of a lowest-level vertex \( u' \in B^{1} \), then \( u \in M_B \) by the algorithm. Removing \( u \) from \( M_B \) would leave \( u' \) without a neighbor in \( M_B \cap B^{all} \) with the same color, again violating the selective subset condition—a contradiction.
\end{itemize}

Hence, this guarantees that \( \lvert M_B \rvert \) is minimized for each block \( B \), and thus \( M \) is a minimum selective subset of \( T \).

\end{proof}

\begin{lemma}\label{runtime}
Algorithm~\ref{alg:02} runs in $O(n)$ time.
\end{lemma}

\begin{proof}
Collecting all sets $B^1$ and $B^{2}$ for all blocks takes $O(n)$ time (Algorithm~\ref{alg:0x}). For each subtree $T^{v_{j}}$ (Algorithm~\ref{alg:02}), at each step we either delete the lowest-level vertex $u$ (if $u \in B^2$), or add a parent or grandparent to $M$ and remove relevant nodes from the sets. Each operation involves only constant-time updates during the update of the sets, and each vertex is processed at most once. Thus, the total running time is $O(n)$.
\end{proof}

\textbf{Remark.} Since our algorithm is described for a single block and applies uniformly to all blocks, the core idea becomes clear: we aim to dominate all vertices in a block that are adjacent to vertices in other blocks. This naturally aligns with an MSOL-expressible formulation. Therefore, the algorithm is not only applicable to trees but also extends to graphs of constant treewidth.

\section{Linear-time Algorithm of MSS in Unit Interval Graphs}\label{section6}
Let \( I = (V(I), E(I)) \) be a unit interval graph with \( |V(I)| = n \). Each vertex in \( V(I) \) corresponds to a unit-length interval placed along the \( x \)-axis. Two vertices \( u \) and \( v \) are adjacent in \( I \) (i.e., \( \{u,v\} \in E(I) \)) if and only if their corresponding intervals overlap. That is, an \emph{edge} exists between two intervals if their projections on the \( x \)-axis intersect. The set \( E(I) \) denotes the collection of all such intersecting pairs of intervals.

Each interval has a \emph{left endpoint} and a \emph{right endpoint} along the \( x \)-axis. We say that an interval \( v \) is a \emph{left adjacent} of another interval \( u \) if the left endpoint of \( v \) is to the left of that of \( u \), and \( v \) intersects \( u \). Conversely, \( v \) is a \emph{right adjacent} of \( u \) if the left endpoint of \( v \) is to the right of that of \( u \), and they intersect. Among all intervals adjacent to a given interval, the one with the smallest \( x \)-coordinate at its left endpoint is referred to as the \emph{leftmost} interval, and the one with the largest \( x \)-coordinate at its right endpoint is called the \emph{rightmost} interval.

Each interval is assigned a color from a set of \( c \) colors. The goal, as in the case of trees (Section~\ref{section5}), is to compute a \mss{} that satisfies the color-consistency condition.

The \textbf{Initialization} phase follows the same strategy discussed for trees, where we sort the intervals by their left endpoints and organize neighborhood relationships accordingly. However, due to the linear geometric layout of unit interval graphs, the adjacency structure becomes more regular and can be exploited for efficiency.

The subsequent steps in the algorithm—such as identifying representative neighbors and choosing selective vertices—are similar in spirit to those for trees, but adapted to the interval model. The key difference lies in how neighbor intervals are processed: instead of using a tree traversal, we perform a left-to-right sweep over the sorted list of intervals, carefully maintaining data structures to track valid neighbor candidates within each color class. This adjustment leverages the geometry of the interval representation to maintain linear-time complexity while preserving correctness.

       

    \textbf{Unit Interval Subgraph Formation.}
    \begin{itemize}
        \item Since $B^{all}\subseteq B$, each $B^{all}$ induces one or more connected induced unit interval subgraphs in $I$, and each such graph consists of vertices of the same color because $B^{all}\subseteq B$. We focus on these induced connected maximal unit interval subgraphs.
        \item Let \( I_{1}, I_{2}, \dots, I_{t} \) be induced connected maximal unit interval subgraphs in $I$ formed by the vertices of \( B^{all} \).
    \end{itemize}

    \textbf{Selection Process.}
    \begin{itemize}         
      \item For each unit interval subgraph $I_{j}$ (where $1\leq j\leq t$), start from the leftmost interval $u$.
      \item $u$ must belong to either $B^1$ or $B^2$.
      \item If $u\in B^2$:
           \begin{itemize}
               \item Remove $u$ from $B^{2}$, $B^{all}$, and $I_{j}$, as it must be adjacent to a interval in $B^{1}$. 
           \end{itemize}
       
       \item If $u\in B^1$:
           \begin{itemize}
               \item If the rightmost adjacent interval of $u$ exists (say $v$) in $I_{j}$, add $v$ to $M_B$ and remove $w$ along with all of its left adjacent vertices from $B^1$, $B^2$, $B^{all}$, and $I_{j}$. Also move all the right adjacent intervals (if exists) of $v$ from $B^{1}$ to $B^2$ and update the sets \( B^{all} \) accordingly.
               \item If $u$ has no right adjacent interval in $I_{j}$, add $u$ to $M_B$ and remove it from $B^2$, $B^{all}$ and $I_{j}$.
           \end{itemize}
        \item Repeat the \textbf{Selection Process} until \( I_{j} \) becomes empty. Once \( I_{j} \) has no vertices, increment \( j \leftarrow j+1 \) and repeat the \textbf{Selection Process} until \( j = t \).
            \end{itemize}
%
%
%
%
Assume $M=\bigcup_{B\in I}M_B$. Since the above algorithm is very similar to the algorithm for trees, Lemma \ref{treelemma} must also hold for unit interval graphs.
\begin{lemma}\label{intervallemma}
    $M$ is a minimum selective subset of the interval graph $I$.
\end{lemma}  
\begin{proof}
We prove this in two steps: First, we show that \( M \) is a selective subset. Then, we prove that \( M \) is of minimum size among all selective subsets.
    
\textbf{Step 1:} (Selective subset)  
This step is very similar to the set \textbf{Step 1} of Lemma \ref{tree}.

\textbf{Step 2:} (Minimality)  
We now prove that the constructed set \( M \) is a minimum selective subset for the unit interval graph \( I \).

Let us consider a unit interval subgraph \( I_{j} \). Suppose \( u \in V_{\ell} \) is the leftmost interval in \( I_{j} \). If \( u \in B^{2} \), we do not immediately include it in \( M_B \). Instead, we consider its unique rightmost adjacent interval \( v \in B^{1} \). In future iterations, we may select either \( v \) or its own rightmost adjacent interval \( w \), if one exists. This strategy works because the interval \( u \) would still satisfy the selective subset condition through \( v \) or \( w \), provided \( d(u, w) \leq d(u, v') \) for some adjacent vertex \( v' \) of \( v \) with a different color. Hence, including \( u \) directly is not necessary.

Next, consider \( u \in B^{1} \). By Lemma~\ref{treelemma}, \( u \) must either be part of \( M_B \), or one of its adjacent intervals in \( B^{all} \) must be included in \( M_B \) to satisfy the selective subset condition.

To argue minimality, suppose a contradiction that some vertex \( u \in M \) can be removed while still maintaining \( M \) as a valid selective subset. Assume \( u \) was added to \( M \) during the iteration for \( I_{j} \). We examine two scenarios:

\begin{itemize}
    \item If \( u \) was the leftmost vertex at the time of its inclusion in \( M_B \), then \( u \notin B^{2} \); otherwise, it would have been excluded during the process. Hence, \( u \in B^{1} \). By construction, \( u \) had no right adjacent vertex (say \( v \)) in \( I_{j} \); otherwise, \( v \) could have been selected instead. If we now exclude \( u \) from \( M_B \), it will have no adjacent vertex in \( M_B \cap B^{all} \) with the same color, thereby violating the selective subset condition—a contradiction.

    \item If \( u \) was later identified as the rightmost neighbor of a leftmost vertex \( u' \in B^{1} \), then \( u \) must be in \( M_B \) according to the algorithm. Removing it would leave \( u' \) without a valid neighbor of the same color in \( M_B \), again violating the condition—a contradiction.
\end{itemize}

Therefore, the selection process ensures that no unnecessary vertices are included in \( M_B \), making it minimal for each block \( B \). As a result, the final set \( M \) is a minimum selective subset of the entire interval graph \( I \).

\end{proof}

The runtime is also the same as explained in Lemma \ref{runtime}.

\section{Remarks}
Since the \mss{} problem is proven to be \(\npc\) even in planar graphs, this raises several natural and important directions for further investigation. One key avenue is the development of efficient approximation algorithms that provide provable guarantees on solution quality. Understanding the approximation limits—both upper and lower bounds—would help delineate the tractability landscape of the problem on planar and other restricted graph classes.

Another promising line of research lies in the study of the parameterized complexity of \mss{}, particularly when the number of colors \(c\) is treated as a parameter. Designing fixed-parameter tractable (\fpt) algorithms when of \(c\) is a parameter, or proving parameterized hardness results, would enhance our theoretical understanding and could lead to practically efficient algorithms for instances where \(c\) is small.

Moreover, exploring the complexity of \mss{} on additional graph classes—such as circular-arc graphs, chordal graphs, and permutation graphs—could uncover new tractable cases or reveal deeper structural insights. These graph families often arise in scheduling, bioinformatics, and network analysis, and understanding \mss{} within these domains may have practical implications as well.

\section*{Acknowledgments}
I would like to sincerely thank Bodhayan Roy for his constant guidance and support throughout this work. I am also grateful to Aritra Banik, Anil Maheshwari, Subhas C Nandy and Sasanka Roy for their valuable suggestions and feedback that helped improve the initial version of this paper. Their insights were instrumental during the early stages of this research. I also thank the anonymous reviewers for their thoughtful comments and suggestions, which helped enhance the quality and presentation of the paper.
\bibliographystyle{unsrt} 
\bibliography{references}  
\end{document}